\documentclass[letterpaper, 10 pt, conference]{ieeeconf}  
\usepackage{graphicx}
\usepackage{epstopdf}
\usepackage{amsmath,amsfonts,amssymb}
\usepackage{algorithmic}
\usepackage[ruled,vlined]{algorithm2e}
\usepackage{marvosym}
\usepackage{mathrsfs}
\usepackage{cite}
\usepackage{array}
\usepackage{soul}
\usepackage{color}
\usepackage{graphics} 
\usepackage{epsfig} 
\usepackage{mathptmx} 
\usepackage{times} 
\usepackage{amsmath} 
\usepackage{amssymb}  
\usepackage{hyperref}
\usepackage{float}
\usepackage{multirow}
\usepackage{paralist, tabularx}
\usepackage[dvipsnames]{xcolor}
\usepackage{soul}
\usepackage[normalem]{ulem}
\newtheorem{remark}{Remark}
\newtheorem{definition}{Definition}
\newtheorem{assm}{Assumption}
\newtheorem{lemma}{Lemma}
\IEEEoverridecommandlockouts                              

\title{\LARGE \bf
Safe Distributed Learning-Enhanced Predictive Control for Multiple Quadrupedal Robots

}

\author{Weishu Zhan$^{1}$, Zheng Liang$^2$, Hongyu Song$^{1}$ and Wei Pan$^{1}$
\textit{}
\thanks{Weishu Zhan, Hongyu Song and Wei Pan are with the Department of Computer Science, The University of Manchester, Manchester M13 9PL, United Kingdom. Zheng Liang is with Zhishen Tech Innovation, China.
}
}
\begin{document}

\maketitle

\thispagestyle{empty}
\pagestyle{empty}

\begin{abstract}
Quadrupedal robots exhibit remarkable adaptability in unstructured environments, making them well-suited for formation control in real-world applications. However, keeping stable formations while ensuring collision-free navigation presents significant challenges due to dynamic obstacles, communication constraints, and the complexity of legged locomotion. This paper proposes a distributed model predictive control framework for multi-quadruped formation control, integrating Control Lyapunov Functions to ensure formation stability and Control Barrier Functions for decentralized safety enforcement. To address the challenge of dynamically changing team structures, we introduce Scale-Adaptive Permutation-Invariant Encoding (SAPIE), which enables robust feature encoding of neighboring robots while preserving permutation invariance. Additionally, we develop a low-latency Data Distribution Service-based communication protocol and an event-triggered deadlock resolution mechanism to enhance real-time coordination and prevent motion stagnation in constrained spaces. Our framework is validated through high-fidelity simulations in NVIDIA Omniverse Isaac Sim and real-world experiments using our custom quadrupedal robotic system, XG. Results demonstrate stable formation control, real-time feasibility, and effective collision avoidance, validating its potential for large-scale deployment. Code are available at the provided link \url{https://github.com/YONEX4090/MultiQuadruped-Control}.


\end{abstract}

\section{Introduction}

Quadruped robots, with their inherent redundant degrees of freedom, possess the capability to adapt their gait patterns to navigate diverse terrains, making them exceptionally suitable for unstructured environments.
Formation control for multiple quadrupedal robotic systems offers significant advantages in applications within these challenging domains (e.g., disaster response \cite{tian2020search}, material transport \cite{pang2012agent}, warehouse management \cite{li2023double}), when compared to alternative multirobot system (MRS) platforms.
However, keeping stable formations is challenging due to terrain variability and dynamic obstacles.
Real-world deployment is further complicated by communication constraints and inefficiencies in data distribution protocols.
To the best of our knowledge, the literature contains only a limited number of studies addressing formation control for multiple quadrupedal robots, largely due to these significant technical challenges \cite{liu2024optimization}.

Traditional approaches, such as leader-follower formations \cite{xiao2016formation} and virtual structures \cite{zhou2018agile}, perform well in structured environments but struggle in highly dynamic scenarios. Model Predictive Control (MPC) is 
a promising method for ensuring the safety and stability of swarm formation control, widely adopted in vehicles and quadrotors, as well as the zero-shot generalizability (i.e., the ability to apply the same controller to new scenarios without additional training).
Its decentralized variant (DMPC) \cite{ferranti2022distributed} enables local decision-making while maintaining global formation stability. However, the high computational cost of MPC hinders real-time scalability.  Control Barrier Functions (CBFs) \cite{ames2019control} offer strong safety guarantees but can be overly conservative, sometimes leading to deadlocks in multi-agent scenarios \cite{wang2017safety, qin2021learning}. Reinforcement learning (RL)-based controllers \cite{zhao2021usv} provide adaptability but often require extensive data to pre-train the policy and often lack formal safety assurances, making them risky for real-world deployment. 
\begin{figure}
    \centering
    \vspace{6pt}

    \includegraphics[width=1\linewidth]{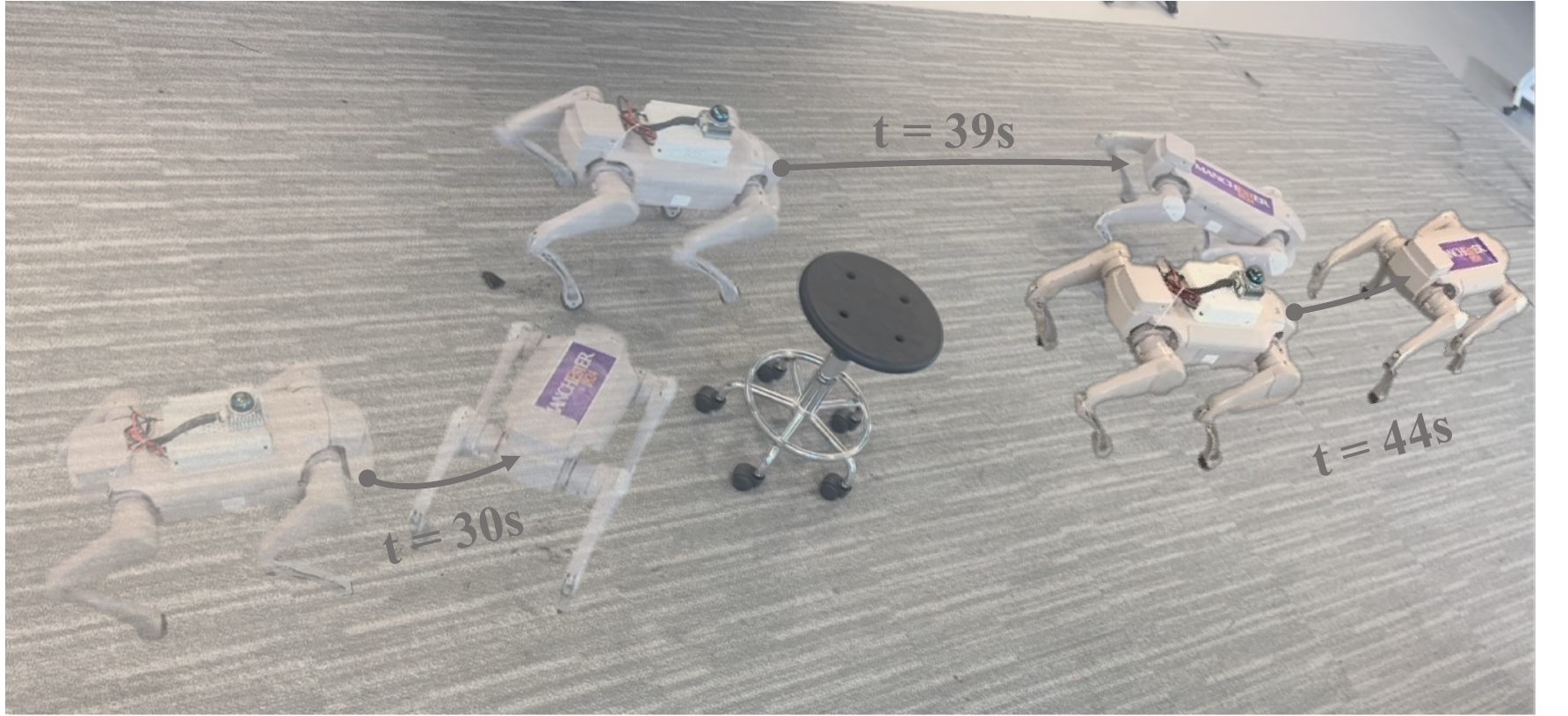}
    \vspace{-14pt}
    \caption{Time-lapse visualization of our quadrupedal formation control system navigating around an obstacle. Three successive positions (t=30s, t=39s, t=44s) show how the robots keeping formation while smoothly avoiding the central obstacle, demonstrating the effectiveness of our distributed control framework in real-world environments.}
    \vspace{-18pt}
    \label{fig:enter-label}
\end{figure}

\vspace{-0.2cm}
\subsection{Related Work}
\vspace{-3pt}
Ren et al. \cite{ren2007information} proposed a consensus-based approach that effectively maintains formation but lacks independent obstacle avoidance capabilities. Zhou et al. \cite{zhou2018agile} utilized artificial potential fields (APFs), modeling obstacles as high-potential regions that generate repulsive forces to steer agents away and ensure collision-free trajectories. However, their reliance on local interactions can lead to deadlocks and limited trajectory optimization. To improve adaptability, hybrid methods have been proposed, such as Qi et al. \cite{qi2022formation}, who integrated consensus theory with pigeon-inspired obstacle avoidance, and approaches combining deep reinforcement learning (DRL) with stochastic braking \cite{zhao2021usv}. RL approaches require extensive training data and struggle to enforce safety, limiting real-world applicability.

MPC has emerged as a powerful tool for multi-robot coordination, balancing formation control and obstacle avoidance through constrained optimization. Zheng et al. \cite{zheng2016distributed} employ DMPC for platoon control under unidirectional communication, while Ferranti et al. \cite{ferranti2022distributed} propose a nonconvex alternating direction method for trajectory optimization. Lyapunov-based DMPC frameworks \cite{wei2019distributed} enhance disturbance handling, yet computational complexity remains a major challenge for large-scale multi-robot systems. Traditional nonlinear solvers suffer from high computational overhead, particularly with long prediction horizons and limited onboard resources. To address this, explicit MPC \cite{alonso2020explicit} precomputes control laws offline, improving real-time efficiency at the cost of exponential offline complexity and strict model accuracy requirements. More recently, learning-based policy refinements have introduced adaptive optimization techniques, reducing computational burdens while preserving closed-loop optimality.

For quadruped robot formations, the combination of nonlinear dynamics and stability constraints complicates coordination. Traditional stability guarantees via Lyapunov functions and CBFs \cite{prajna2006barrier} are challenging to derive for legged locomotion, while precise foot placement and force control remain essential for unstructured terrain traversal \cite{arena2021learning}. RL and hybrid methods \cite{westenbroek2022lyapunov} integrate Zero Moment Point and iterative learning control \cite{hu2016learning} to improve stability, but real-world safety remains difficult to guarantee. Computationally intensive methods such as MPC often struggle with real-time execution due to onboard constraints, whereas neural controllers shift computations offline, enabling efficient deployment with integrated stability and safety guarantees. Although \cite{xu2023distributed, liu2024optimization} have explored multi-quadruped robot obstacle avoidance, gait synchronization, and flocking control, they lack theoretical guarantees on formation error convergence and do not fully address the integration of formation control and collision avoidance.

Our framework integrates MPC-CLF for stability, SAPIE for adaptive neighborhood encoding, and a deep neural network-based CBF for decentralized safety, ensuring real-time feasibility with formal safety guarantees. Traditional methods lack explicit obstacle avoidance, encounter deadlocks, or impose high computational costs, while our approach enforces safety constraints, ensures motion feasibility in dense formations, and efficiently encodes dynamic neighbor relationships. By reducing the computational burden of DMPC and explicitly guaranteeing formation error convergence, our approach provides a scalable, efficient, and robust solution for quadrupedal formation control in complex, dynamic environments.

\vspace{-10pt}
\subsection{Contribution}
\vspace{-5pt}
To ensure the safe navigation of quadruped robot formations, this paper presents a safe distributed learning-enhanced predictive control framework. Our contributions are as follows:

\begin{itemize}
\item We propose a distributed model predictive control framework for formation control and collision avoidance of multiple quadrupedal robots. We provide a theoretical proof that the formation error can converge to zero.

\item We introduce a Scale-Adaptive Permutation-Invariant Encoding (SAPIE) mechanism to address feature dimensional mismatches in multi-quadrupedal robots. Furthermore, safety and forward invariance can be ensured by combining SAPIE with deep neural networks and CBF.
\item We develop a custom quadrupedal robotic system platform, XG, which incorporates a low-latency Data Distribution Service protocol with a Peer-to-Peer architecture to enable reliable, real-time communication. An event-triggered mechanism further detects and resolves deadlocks in narrow passages or complex terrains, improving overall coordination.
\item We develop a digital twin in NVIDIA Omniverse Isaac Sim that closely replicates real-world properties. By deploying our custom-designed quadruped robots and communication topology, we conducted both simulation and real-world experiments to validate the effectiveness and robustness of the proposed controller in collision avoidance tasks, as shown in Fig. \ref{fig:enter-label}.
\end{itemize}
\section{Problem Formulation}
\label{sec:2}
This section presents the hierarchical control framework, the control-affine system for quadruped robots, and key safety concepts, including decentralized CBFs. We then formalize the formation control problem and outline key assumptions.

\vspace{-0.2cm}
\subsection{Quadruped Robot Model}

The control-affine system for a typical model of the $i$-th quadruped robot as follows:
\vspace{-0.2cm}
\begin{align}
\label{Eq. 1}
\dot{\boldsymbol{x}}_i&=f(\boldsymbol{x}_i)+g(\boldsymbol{x}_i) \boldsymbol{u}_i 
\\
\boldsymbol{o}_i&=\left\{\boldsymbol{x}_j \mid j \in \mathcal{N}_i\right\}
\end{align}
where $\boldsymbol{x}_i \in \mathbb{X}_i \subseteq \mathbb{R}^n$ and $\boldsymbol{u}_i \in \mathbb{U}_i \subseteq \mathbb{R}^m$ is the state vector and the input of the $i$-th quadruped robot. For each agent $i$, We define $\mathcal{N}_i$ as the set of neighboring robots, and the local
 observation $\boldsymbol{o}_i \in \mathbb{R}^{n \times\left|\mathcal{N}_i\right|}$ consists of the states of $\mathcal{N}_i$  neighborhood robots. $f$ and $g$ are locally continuous Lipschitz functions.
\begin{remark}
    The whole-body dynamics of a quadruped robot are highly complex, making DMPC implementation challenging. To address this, we adopt a hierarchical control structure, where DMPC functions as the high-level controller, while Walk-These-Ways (WTW) \cite{margolis2023walk} simulates low-level motor control. WTW learns motion behaviors with configurable posture, gait, and velocity, replicating built-in controller parameters. During high-level training, parallel GPU execution enhances computational efficiency, ensuring real-time DMPC optimization. 
\end{remark}

\textcolor{black}{Irregularities in system dynamics can compromise stability and the feasibility of control synthesis. To ensure well-posedness and the existence of solutions, it is necessary to establish the regularity of the functions 
$f$ and $g$. In particular, we impose local Lipschitz continuity, a fundamental property that ensures the bounded variation of system dynamics under small perturbations in the state.}

\begin{definition}
A function $f: \mathbb{R}^n \rightarrow \mathbb{R}^n$ is locally Lipschitz continuous at $\boldsymbol{x}_i \in \mathbb{X}_i$ if there exist constants $M>0$ and $\delta>0$ such that, for all $\boldsymbol{x}_i, \boldsymbol{x}_i^{\prime}$ satisfying $\left\|\boldsymbol{x}_i-\boldsymbol{x}_i^{\prime}\right\| \leq \delta:$
$\left\|f\left(\boldsymbol{x}_i\right)-f\left(\boldsymbol{x}_i^{\prime}\right)\right\| \leq M\left\|\boldsymbol{x}_i-\boldsymbol{x}_i^{\prime}\right\|, 
\left\|g\left(\boldsymbol{x}_i\right)-g\left(\boldsymbol{x}_i^{\prime}\right)\right\| \leq M\left\|\boldsymbol{x}_i-\boldsymbol{x}_i^{\prime}\right\|.
$
\hfill $\blacksquare$

\end{definition}

\begin{figure}[t]
    \centering
    \includegraphics[width=\columnwidth]{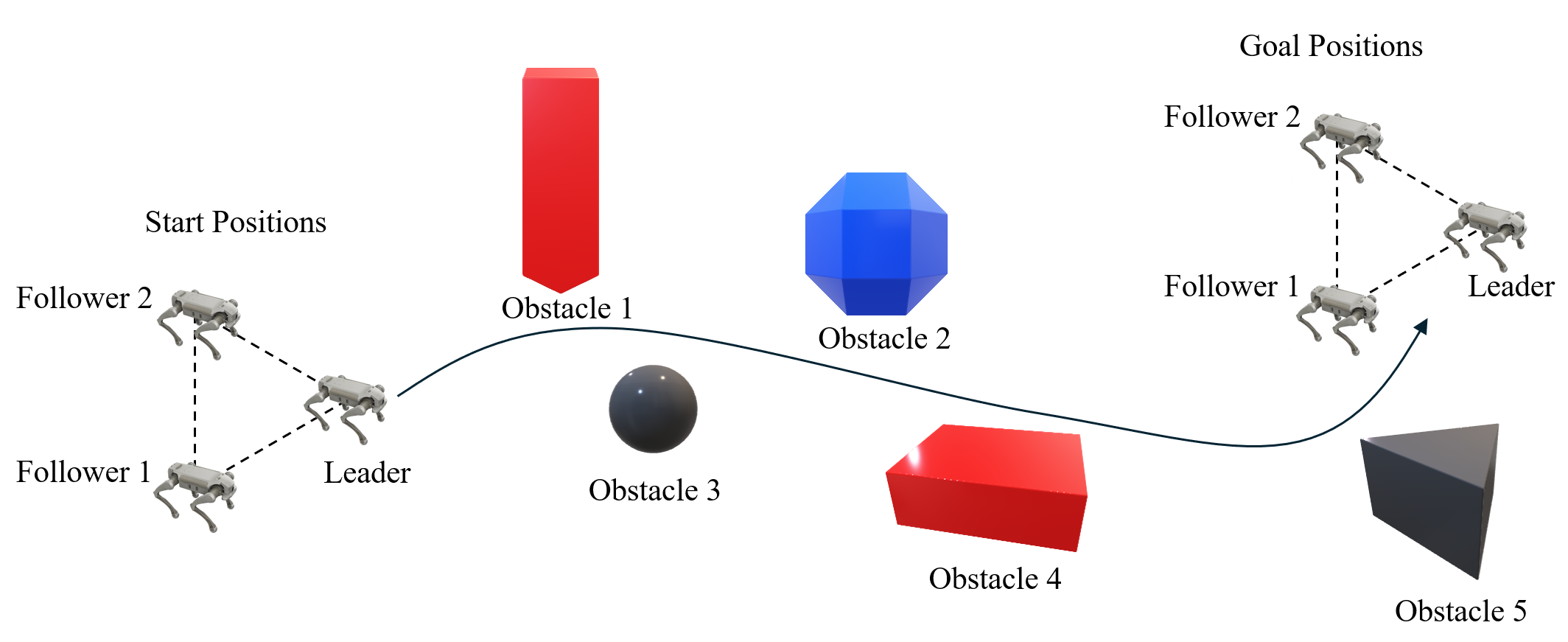}
    \vspace{-20pt}
    \caption{Quadruped robots cooperative formation and collision avoidance.}
    \vspace{-0.7cm}
    \label{fig2}
\end{figure}

 \subsection{Problem Statement}

In this paper, safe control refers to devising control inputs to guarantee that: i) states of the dynamic system never reach the unsafe region; ii) control inputs never violate the physical constraints of the system. We consider the formation control problem for a multi-agent quadruped robot system operating in an unknown environment, as shown in Fig. \ref{fig2}. The objective is to navigate the robots from their initial positions to designated goal positions while maintaining formation and ensuring collision-free movement with both obstacles and neighboring robots. Each quadruped robot perceives obstacles using onboard sensors and communicates with nearby robots to exchange state information.

\begin{definition}
A continuous function $\alpha:[0, a) \rightarrow[0, \infty)$ belongs to class $\kappa$ if it is strictly increasing and satisfies $\alpha(0)=0$. An extended class $\kappa$ function $\beta: \mathbb{R} \rightarrow \mathbb{R}$ is also strictly increasing and satisfies $\beta(0)=0$.\hfill $\blacksquare$
\end{definition}

\begin{definition}
The Lie derivative of a function $\eta(\boldsymbol{x})$ along a vector field $\xi(\boldsymbol{x})$ is defined as:
$
L_{\xi} \eta(\boldsymbol{x})=\frac{\partial \eta(\boldsymbol{x})}{\partial \boldsymbol{x}} \xi(\boldsymbol{x}).
$
\hfill $\blacksquare$
\end{definition}

\begin{definition}[CBF\cite{ames2016control}]
A continuously differentiable function $h: \mathbb{R}^n \rightarrow \mathbb{R}$ is a CBF if it defines a safe set:
$
{S}:=\{\boldsymbol{x} \mid h(\boldsymbol{x}) \geq 0\}
$
and there exists an extended class $\kappa$ function 
$\alpha(\cdot): \mathbb{R} \rightarrow \mathbb{R}$ such that:
$
\dot{h}(\boldsymbol{x}, \boldsymbol{u})+\kappa(h(\boldsymbol{x})) \geq 0, \forall \boldsymbol{x} \in \mathscr{H}
$.
\hfill $\blacksquare$
\end{definition}

Based on the above definitions, we give the definition on safety for multiple quadrupedal robots as follows:  
\begin{definition}\label{definition:cbfloss}
A multi-agent system \eqref{Eq. 1} is safe at time $t$  if each agent $i$ satisfies:
$$
\boldsymbol{x}_i(t) \in \mathbb{X}_{safe, i}, \quad \boldsymbol{u}_i(t) \in \mathbb{U}_{safe, i}, \quad \forall t \geq 0
$$
where $\mathbb{X}_{safe, i}$ and $\mathbb{U}_{safe, i}$ denote the safe state and safe control input sets for agent $i$, respectively. The system maintains safety if the minimum safety distance s between agents satisfies:
$
d\left(\boldsymbol{x}_i, \boldsymbol{o}_i\right) \geq s, \quad \forall i. 
$
If the safe set $S \subseteq \mathbb{X}_{safe,i}$ is forward invariant, meaning:
$$
\boldsymbol{x}_i(0) \in S \Rightarrow \boldsymbol{x}_i(t) \in S, \quad \forall t \geq 0
$$
then safety is preserved over time. To ensure this, we introduce a decentralized CBF to enforce the forward invariance of the safe set.
A decentralized CBF $h_i: \mathbb{X}_i \rightarrow \mathbb{R}$ defines the zero-superlevel safe set:
$$
\begin{aligned}
S_i=\left\{\boldsymbol{x}_i \mid h_i\left(\boldsymbol{x}_i, 
\boldsymbol{o}_i\right) \geq 0\right\}
\\
\partial S_i=\left\{\boldsymbol{x}_i \mid h_i\left(\boldsymbol{x}_i, \boldsymbol{o}_i\right) = 0\right\}
\\
\operatorname{Int}\left(S_i\right)=\left\{\boldsymbol{x}_i \mid h_i\left(\boldsymbol{x}_i, \boldsymbol{o}_i\right) > 0\right\}
\end{aligned}
$$
where $\partial S_i$ is the boundary set, and $\operatorname{Int}\left(S_i\right)$ is the interior set, where $h_i\left(\boldsymbol{x}_i, \boldsymbol{o}_i\right)>0$.
To ensure $S_i$ is forward invariant, $h_i$ must satisfy the CBF condition:
$$
\sup _{\boldsymbol{u}_i \in \mathbb{U}_i}\left[L_{f} h_i\left(\boldsymbol{x}_i, \boldsymbol{o}_i\right)+L_{g} h_i\left(\boldsymbol{x}_i, \boldsymbol{o}_i\right) \boldsymbol{u}_i\right] \geq-\alpha\left(h_i\left(\boldsymbol{x}_i, \boldsymbol{o}_i\right)\right)
$$
where $\alpha(\cdot)$ is an extended class $\kappa$ function, and $L_{f}, L_{g}$ are Lie derivatives. If this condition holds, $S_i$ remains forward invariant, ensuring individual agent safety and guaranteeing the safety of the entire multi-agent system.\hfill $\blacksquare$
\end{definition}

\textcolor{black}{Since our work focuses on designing a fast policy learning algorithm for DMPC, we introduce two essential assumptions to ensure the controllability of formation control.}
\begin{assm}
    The communication network is time-invariant and delay-free, meaning the topology remains unchanged during formation maintenance and obstacle avoidance. Neighboring robots share state information instantaneously. Each robot has complete awareness of its neighbors' states. 
\end{assm}


\begin{assm}
    Each quadruped robot $i$ has a local feedback control policy $\boldsymbol{u}_i$ that stabilizes its formation tracking error while ensuring safety. This policy does not depend on a specific communication topology; information exchange between neighboring robots can be either bidirectional or unidirectional, provided stability is maintained. We assume there exists a feedback policy (which we aim to obtain via our DMPC approach) that stabilizes formation error to ensure both formation stability and obstacle avoidance.
\end{assm}

\section{DISTRIBUTED CONTROLLER DESIGN}
\label{sec:3}


The objective of this section is to present a distributed control framework for formation keeping and collision avoidance of quadruped robots. We integrate MPC with CLFs to ensure stable formation tracking while incorporating CBFs to enforce safety constraints. A Neural Network-based CBF architecture is introduced to enhance adaptability to dynamic multi-agent interactions, addressing challenges posed by varying agent densities and environmental complexities. Additionally, we propose an event-triggered mechanism to resolve deadlocks that arise when nominal control inputs conflict with CBF constraints, ensuring continuous motion without compromising stability. The following subsections detail the MPC formulation with CLF constraints, the learning-based CBF design, and the deadlock resolution strategy. Fig.~\ref{fig:framework} shows the distributed control framework

\begin{figure}[t]
	\centering
	\includegraphics[width=\columnwidth]{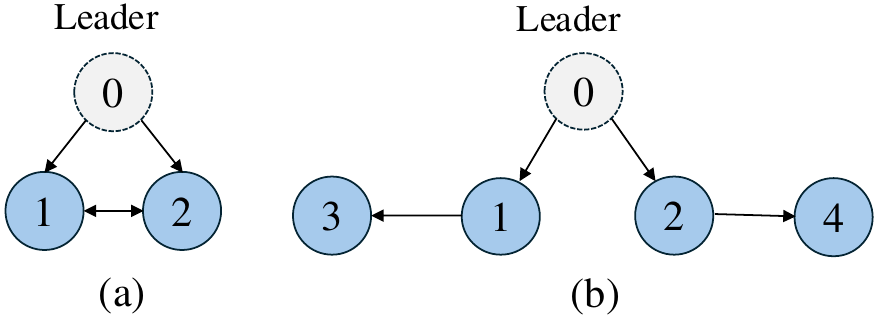}
    \vspace{-0.6cm}
	\caption{
    Exemplary scenarios of communication graphs with (a) N = 3 and (b) N = 5. The arrows represent the directions of information exchange among robots. Communications are instantaneously exchanged between neighboring robots at each step.}
	\label{fig3}
 \vspace{-21pt}
\end{figure}
\vspace{-3pt}
\subsection{MPC With CLF Constraints}\label{sec:mpc-clf}
\vspace{-3pt}
We develop a distributed MPC framework with CLF constraints, extending its application from single-robot locomotion to multirobot formation control. This approach ensures the preservation of the desired formation and maintains velocity consistency across quadrupedal robots.
 

\begin{figure*}[ht]
	\centering
	\includegraphics[width=\textwidth]{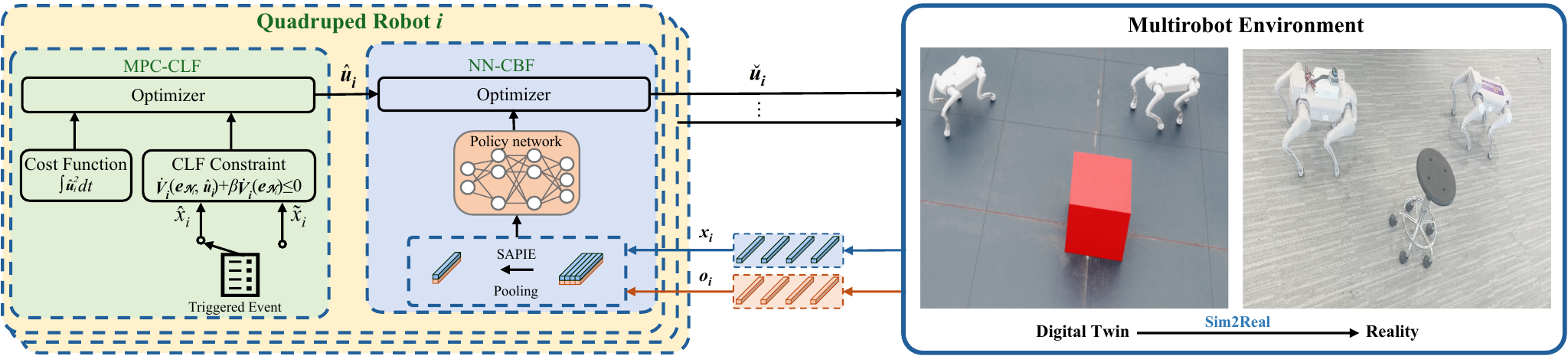}
    \vspace{-0.6cm}
	\caption{Overview of the presented approach. Left: The distributed control architecture for each quadruped robot combines an MPC-CLF optimizer for formation stability with an SAPIE-CBF Neural Networks for safety enforcement. The SAPIE module processes variable neighbor configurations through a permutation-invariant encoding mechanism, while the event-triggered mechanism resolves potential deadlocks. Right: Multirobot environment validation showing the seamless transition from digital twin simulation to physical deployment, demonstrating the effectiveness of our framework in both simulated and real-world scenarios.
    }
    \label{fig:framework}
        \vspace{-0.5cm}
\end{figure*}


In this paper, to reduce the communication load and computation burden, we consider a fully distributed approach: each agent only communicates with its neighbors as exemplified in Fig.~\ref{fig6}. Therefore, we define the nominal state $\hat{\boldsymbol{x}}_i \in \mathbb{R}^4$ computed from the states of its neighboring robots as:
\begin{equation}
\hat{\boldsymbol{x}}_i=\frac{\sum_{j=1, j \neq i}^n c_{i j}\left(x_j+\Delta_{ij}\right)+s_i\left(x_r+\Delta_{i r}\right)}{\sum_{j=1, j \neq i}^n c_{i j}+s_i}
\label{Eq.nominalstate}
\end{equation}
where $\hat{\boldsymbol{x}}_i=\left(\boldsymbol{p}_{x, i}, \boldsymbol{p}_{y, i}, \boldsymbol{\theta_i}, \boldsymbol{v_i}\right) \in \mathbb{R}^4$, with $(\boldsymbol{p}_{x, i}, \boldsymbol{p}_{y, i})$ representing the position of quadruped robot $i$. The parameter $c_{i j}$ indicates the connection status, where $c_{i j}=1$ for $j \in \mathcal{N}_i$ and $c_{i j}=0$ otherwise. Parameter $s_i=1$ indicates that robot $i$ receives position information from the leader. The reference state $x_r$ is obtained from the leader. The terms $\boldsymbol{\Delta}_{i j}$ and $\Delta_{i r}$ represent formation offsets, which are determined by the formation shape and size.

We define the formation error as $\boldsymbol{e}_{\mathcal{N}_i}$=$\boldsymbol{x}_i-\hat{\boldsymbol{x}}_i$. The MPC with CLF constraints is formulated as follows:
\begin{equation}
\begin{aligned}
\min \ &  J_{\mathrm{CLFs}}\left(\hat{\boldsymbol{u}}_i\right)=\int_t^{t+T} \hat{\boldsymbol{u}}_i^2 d t \\
\text { s.t. } &  \dot{\boldsymbol{x}}_i=f(\boldsymbol{x}_i)+g(\boldsymbol{x}_i) \hat{\boldsymbol{u}_i} \\
& \dot{V}_i\left(\boldsymbol{e}_{\mathcal{N}_i}, \hat{\boldsymbol{u}}_i\right)+\beta V_i\left(\boldsymbol{e}_{\mathcal{N}_i}\right) \leq 0 \\
& \boldsymbol{v}_{\min } \leq \boldsymbol{v}_i \leq \boldsymbol{v}_{\max } \\
& \boldsymbol{u}_{\min } \leq \boldsymbol{u}_i \leq \boldsymbol{u}_{\max }
\end{aligned}
\label{Eq. 3}
\end{equation}
where $V_i\left(\boldsymbol{e}_{\mathcal{N}_i}\right)=\boldsymbol{e}_{\mathcal{N}_i}^T \boldsymbol{e}_{\mathcal{N}_i}$ is the Lyapunov function whose properties are given in the following Definition~\ref{definition:clf}; $\hat{\boldsymbol{u}}$ is the control input; $\boldsymbol{v}_{\min }$, $\boldsymbol{v}_{\max }$, $\boldsymbol{u}_{\min }$, and $\boldsymbol{u}_{\max }$ denote the upper and lower bounds of $\boldsymbol{v}_i$ and $\boldsymbol{u}_i$ respectively.
\begin{definition}[\cite{minniti2021adaptive}]\label{definition:clf}

    A continuously differentiable function $V: \mathbb{X} \rightarrow \mathbb{R}$ is considered a CLF for a quadruped robot if there exist extended class $\kappa$ functions $\beta_1, \beta_2, \beta_3$ such that for all $\boldsymbol{e}_{\mathcal{N}_i} \in \mathbb{X}$ :

(i) Positive Definiteness:
\vspace{-4pt}
$$
\beta_1\left(\left\|\boldsymbol{e}_{\mathcal{N}_i}\right\|\right) \leq V\left(\boldsymbol{e}_{\mathcal{N}_i}\right) \leq \beta_2\left(\left\|\boldsymbol{e}_{\mathcal{N}_i}\right\|\right)
$$
\vspace{-4pt}
where $V\left(\boldsymbol{e}_{\mathcal{N}_i}\right)=0$ if and only if $\boldsymbol{e}_{\mathcal{N}_i}=0$, and $V\left(\boldsymbol{e}_{\mathcal{N}_i}\right)>0$ otherwise.

(ii) Exponential Stability Constraint:
\vspace{-7pt}
$$
\begin{aligned}
\inf _{\boldsymbol{u}_i \in \mathbb{U}}\Big[\dot{V}\left(\boldsymbol{e}_{\mathcal{N}_i}, \boldsymbol{u}_i\right)
&=L_f V\left(\boldsymbol{e}_{\mathcal{N}_i}\right)+L_g V\left(\boldsymbol{e}_{\mathcal{N}_i}\right) \boldsymbol{u}_i\Big]
\leq-\beta_3\left(V\left(e_{\mathcal{N}_i}\right)\right)
\end{aligned}
$$
\vspace{-3pt}
ensuring that the formation error $\boldsymbol{e}_{\mathcal{N}_i}$ converges exponentially to zero under the control policy:

$
\pi_{\mathrm{CLF}}=\left\{\boldsymbol{u}_i \in \mathbb{U} \mid \dot{V}\left(\boldsymbol{e}_{\mathcal{N}_i}, \boldsymbol{u}_i\right) \leq-\beta_3\left(V\left(\boldsymbol{e}_{\mathcal{N}_i}\right)\right)\right\}$.
\hfill $\blacksquare$
\end{definition}

Combining \eqref{Eq. 1} with \eqref{Eq.nominalstate}, we can obtain the desired dynamics for the nominal state $\hat{x}_i$:
\begin{equation}\label{eq:nominaldynamics}
\begin{aligned}
\dot{\hat{\boldsymbol{x}}}_i
=\frac{(\sum_{j=1, j \neq i}^n c_{i j}\left(f\left(\boldsymbol{x}_j\right)+g\left(\boldsymbol{x}_j\right) {\boldsymbol{u}}_j\right))+(s_i\left(f\left(\boldsymbol{x}_r\right)+g\left(\boldsymbol{x}_r\right) {\boldsymbol{u}}_r\right))}{\sum_{j=1, j \neq i}^n c_{i j}+s_i} \\
\end{aligned}
\end{equation}
By using \eqref{eq:nominaldynamics}, we can derive the derivative of the Lyapunov function $V_i\left(\boldsymbol{e}_{\mathcal{N}_i}\right)$ as follows:
\begin{equation}
\begin{aligned}
\dot{V}_i\left(\boldsymbol{e}_i, \boldsymbol{u}_i\right)
&=2 \boldsymbol{e}_i^T\left(\dot{\boldsymbol{x}}_i-\dot{\hat{\boldsymbol{x}}}_i\right)  \\
& =\frac{\partial V_i\left(\boldsymbol{x}_i-\hat{\boldsymbol{x}}_i\right)}{\partial \boldsymbol{x}_i} \dot{\boldsymbol{x}}_i+\frac{\partial V_i\left(\boldsymbol{x}_i-\hat{\boldsymbol{x}}_i\right)}{\partial \hat{\boldsymbol{x}}_i} \dot{\hat{\boldsymbol{x}}}_i \\
&=2\left(\boldsymbol{x}_i -\hat{\boldsymbol{x}}_i\right)^T\left(f\left(\boldsymbol{x}_i\right)+g\left(\boldsymbol{x}_i\right) \boldsymbol{u}_i-\dot{\hat{\boldsymbol{x}}}_i\right)
\end{aligned}
\label{Eq. 4}
\end{equation}
Building upon the derivation in \eqref{Eq. 4}, we have the following Lemma to show the formation error $\boldsymbol{e}_{\mathcal{N}_i}$=$\boldsymbol{x}_i-\hat{\boldsymbol{x}}_i$ can converge to 0.
\begin{lemma}
    When the optimal cost function $J_{\text {CLFs }}\left(\hat{\boldsymbol{u}}_i\right)$ reaches zero, the CLF constraint becomes slack, indicating that the dynamics of the system inherently drives the formation error $\boldsymbol{e}_{\mathcal{N}_i}$ to zero, ensuring stability without additional control effort. 
\end{lemma}

\begin{proof}
  To prove this Lemma, we analyze the MPC-CLF framework using the duality theorem and the Karush-KuhnTucker (KKT) conditions \cite{boyd2004convex}. We begin by formulating the Lagrangian:
$$
\begin{aligned}
\mathbb{L}\left(\boldsymbol{x}_i, \hat{\boldsymbol{u}}_i, \lambda_i, \mu_i\right)= & J_{\mathrm{CLFs}}\left(\hat{\boldsymbol{u}}_i\right) \\
& +\lambda_i\left(\dot{V}_i\left(\boldsymbol{x}_i-\hat{\boldsymbol{x}}_i, \hat{\boldsymbol{u}}_i\right)+\beta V_i\left(\boldsymbol{x}_i-\hat{\boldsymbol{x}}_i\right)\right) \\
& +\mu_i\left(\dot{\boldsymbol{x}}_i-f(\boldsymbol{x}_i)+g(\boldsymbol{x}_i)\hat{\boldsymbol{u}}_i\right)
\end{aligned}
$$
where $\lambda_i$ and $\mu_i$ are the Lagrange multipliers.
Let $\left(\boldsymbol{x}_i^*, \boldsymbol{u}_i^*, \lambda_i^*, \mu_i^*\right)$ be the optimal primal-dual solution. The KKT conditions are as follows.
\begin{enumerate}
    \item Stationarity:
$
\left.\nabla_{\left(\boldsymbol{x}_i, \boldsymbol{u}_i\right)} L\left(\boldsymbol{x}_i, \boldsymbol{u}_i, \boldsymbol{\lambda}_i, \boldsymbol{\mu}_i\right)\right|_{\left(\boldsymbol{x}_i^*, \boldsymbol{u}_i^*, \boldsymbol{\lambda}_i^*, \boldsymbol{\mu}_i^*\right)}=0
$
\item Primal Feasibility:
$
\begin{aligned}
& \dot{\boldsymbol{x}}_i^*=f\left(x_i\right)+g\left(x_i\right) \boldsymbol{u}_i^* \\
& \dot{V}_i\left(\boldsymbol{x}_i^*-\hat{\boldsymbol{x}}_i, \boldsymbol{u}_i^*\right)+\beta V_i\left(\boldsymbol{x}_i^*-\hat{\boldsymbol{x}}_i\right) \leq 0
\end{aligned}
$
\item Dual Feasibility:
$
\lambda_i^* \geq 0
$
\item Complementary Slackness:
$$
\lambda_i^*\left(\dot{V}_i\left(\boldsymbol{x}_i^*-\hat{\boldsymbol{x}}_i, \boldsymbol{u}_i^*\right)+\beta V_i\left(\boldsymbol{x}_i^*-\hat{\boldsymbol{x}}_i\right)\right)=0.
$$
We assume 
$
\lambda_i^*\left(\dot{V}_i\left(\boldsymbol{x}_i^*-\hat{\boldsymbol{x}}_i, \boldsymbol{u}_i^*\right)+\beta V_i\left(\boldsymbol{x}_i^*-\hat{\boldsymbol{x}}_i\right)\right)<0
$
According to the complementary slackness condition in KKT, it it follows that $\lambda_i^*=0$. 
$
\begin{aligned}
\mathbb{L}\left(\boldsymbol{x}_i, \hat{\boldsymbol{u}}_i, \lambda_i, \mu_i\right)= & J_{\mathrm{CLFs}}\left(\hat{\boldsymbol{u}}_i\right) +\mu_i\left(\dot{\boldsymbol{x}}_i-f(\boldsymbol{x}_i)+g(\boldsymbol{x}_i)\hat{\boldsymbol{u}}_i\right)
\end{aligned}
$, which is equivalent to MPC without CLF constraints.
\end{enumerate}

Furthermore, the optimal control solution in this case is $\boldsymbol{u}_i^*=\mathbf{0}$. Substituting this into Eq. \ref{Eq. 3} shows that Lyapunov's second law can only be satisfied when $\boldsymbol{x}_i^*=\hat{\boldsymbol{x}}_i$, and consequently, $J_{\mathrm{CLFs}}\left(\hat{\boldsymbol{u}}_i\right)$ equals zero.  
\end{proof}
This is particularly important for quadrupedal robots, as it establishes a theoretically sound stability condition despite their complex, high-dimensional dynamics and gait variability.
\vspace{-8pt}
\subsection{SAPIE for Dynamic Neighborhood Adaptation}
\vspace{-5pt}
\label{sec:sapie}
In practical multiple quadrupedal robotic systems, we often encounter scalability and permutation invariance problems. Therefore, effectively processing dynamically changing neighbor sets while ensuring safety constraints remain valid is a critical challenge (The control law \cite{cavorsi2023multirobot}
$
u_i=\sum_{j \in \mathcal{N}_i(t)}\left(f_{\text {repel }}\left(p_i-p_j\right)+f_{\text {align }}\left(v_i, v_j\right)\right)
$
where $\mathcal{N}_i(t)$ denotes the set of neighbors of robot $i$ at time $t, f_{\text {repel }}\left(p_i-p_j\right)$ represents a repulsive force function to prevent collisions, $f_{\text {align }}\left(v_i, v_j\right)$ captures an alignment force to ensure velocity consistency, and $u_i$ is the control input determined by the cumulative influence of all current neighbors.

This formulation necessitates real-time adaptation of $u_i$ as neighbors dynamically join or leave $\mathcal{N}_i(t)$.
While the summation structure guarantees permutation invariance, ensuring stability and collision avoidance in an evolving topology remains nontrivial.) Traditional control methods often assume a fixed number of neighbors, but in real-world deployments, neighbor count fluctuates, and their order changes arbitrarily. This variability introduces two key challenges: scalability and permutation invariance. If a control strategy fails to adapt to dynamic neighbor variations or produces inconsistent outputs based on neighbor ordering, it risks violating safety constraints, compromising system stability and reliability.

To address these challenges, we propose a new Scale Adaptive Permutation-Invariant Encoding (SAPIE) mechanism. SAPIE enables quadrupedal robots to efficiently encode dynamic neighbor information while ensuring invariance to neighbor order in control policies. SAPIE takes each neighbor’s state, maps it through a small network $\phi_j$, and then uses an order-independent pooling (max) so that whether a robot has 5 or 10 neighbors, the output is a fixed-length representation. The core idea is to apply shared nonlinear transformations to map neighbor states into fixed-length feature vectors and normalize them through symmetric pooling operations, ensuring robustness to variations in neighbor count and order.
Mathematically, the state of each neighboring agent is transformed using a shared nonlinear function:
$
\phi_i = \sigma(W \boldsymbol{\boldsymbol{o}_i}),
$
where $W \in \mathbb{R}^{p \times n}$ is a weight matrix, and $\sigma(\cdot)$ is a nonlinear activation function. Since the number of neighbors $\left|\mathcal{N}_i(t)\right|$ varies, SAPIE aggregates all neighbor feature vectors using row-wise pooling to maintain permutation invariance:
\vspace{-5pt}
\begin{equation}
\begin{aligned}
\psi(\boldsymbol{o}_i) = \text{SAPIE}(\phi_1, \phi_2, \dots, \phi_{|\mathcal{N}_i(t)|})
\end{aligned}
\label{Eq.sapie}
\end{equation}
which maps a dynamically sized input $[\phi_1, \phi_2, \dots, \phi_{|\mathcal{N}_i(t)|}]$ to a fixed-length feature vector $\psi(\boldsymbol{o}_i)$. To ensure permutation invariance, SAPIE employs row-wise max pooling, ensuring that the output feature vector remains unchanged regardless of the number or order of neighbors:
$$
\psi(\boldsymbol{o}_i): \mathbb{R}^{n \times\left|\mathcal{N}_i(t)\right|} \mapsto \mathbb{R}^s.
$$
\vspace{-2pt}
We will demonstrate the effectiveness of SAPIE in Scenario 1: Safe Navigation in Obstacle-Dense Environment (see Section~\ref{sec:Scenario 1}). Specifically, we evaluate how SAPIE enables quadrupedal robots to adapt to varying neighbor counts while keeping stable formation control. In our experiments, robots navigate through complex environments with dynamically changing neighbor configurations, ensuring consistent performance despite fluctuations in neighbor numbers and order. The results in Fig.~\ref{fig7} confirm that SAPIE effectively encodes neighbor information, allowing the controller to prioritize safety while preserving formation integrity.
\vspace{-5pt}
\subsection{SAPIE-CBF Based on Neural Networks}\label{sec:nn-cbf}
\vspace{-3pt}
Typically, the current CBF methods cannot be used directly for dynamic multi-agent environments, as we introduced in Section~\ref{sec:sapie}, due to their dependence on fixed neighborhood structures.
To address this, we introduce deep neural network to approximate $\psi(\boldsymbol{o}_i) $ in \eqref{Eq.sapie},
that adapts to varying agent densities while maintaining safety constraints.

Based on Definition~\ref{definition:cbfloss}, the empirical loss function for training the CBF is: 
\begin{equation}
\begin{aligned}
\mathbf{L}_{CBF}=\Sigma_i {L}^i_{CBF}
\end{aligned}
\label{Eq.tatalcbf}
\end{equation}
where we propose a novel loss function ${L}^i_{CBF}$ for each quadruped robot $i$:
\begin{equation}
\begin{aligned}
L_{C B F }^i\left(\theta_i\right)
&=\sum_{\boldsymbol{x}_i \in \mathbb{X}_{\text {safe,i }}} \max \left(0, \gamma-h_i^{\theta_i}\left(\boldsymbol{x}_i, \boldsymbol{o}_i\right)\right)
\\
&+\sum_{\boldsymbol{x}_i \in \mathbb{X}_{\text {unsafe,i}}} \max \left(0, \gamma+h_i^{\theta_i}\left(\boldsymbol{x}_i, \boldsymbol{o}_i\right)\right)
\\
&+\sum_{\boldsymbol{x}_i \in \mathbb{X}_h} \max \Big(0, \gamma-L_{f} h_i^{\theta_i}\left(\boldsymbol{x}_i, \boldsymbol{o}_i\right)
\\
& \ \ \ \ -L_{g} h_i^{\theta_i}\left(\boldsymbol{x}_i, \boldsymbol{o}_i\right) \cdot \boldsymbol{u}_i-\alpha\left(h_i^{\theta_i}\left(\boldsymbol{x}_i, \boldsymbol{o}_i\right)\right)\Big)
\end{aligned}
\label{Eq.lcbf}
\end{equation}
On the right-hand side of \eqref{Eq.lcbf}, the first two terms enforce the safety classification; and the third term ensures forward invariance, i.e., if the system's initial state is within the safe set, all subsequent system states will remain within the safe set.
$\gamma$ is the margin for the satisfaction of the CBF condition in Definition~\ref{definition:cbfloss}. We choose $\gamma=10^{-3}$ in implementation. $\theta_i$ is neural network parameter. $\mathbb{X}_{\text {s,i }}$ is the safe state set, $\boldsymbol{x}_i \in \mathbb{X}_{\mathrm{unsafe}, \mathrm{i}}$ is the unsafe state set, where \( h_i(\boldsymbol{s}_i, \boldsymbol{o}_i) < 0 \). The set $\mathbb{X}_h$ represents the state trajectories where the CBF constraint is evaluated.

In practice, the control input $\hat{u}_i$ which is solved by MPC-CLF~\eqref{Eq. 3} cannot be used directly because without the predefined trajectory solved in \eqref{Eq.lcbf}, the follower will simply terminate the motion in front of obstacles to maintain safety.
Therefore, the control input needs to be refined to
remain close to the original MPC-CLF output $\hat{\boldsymbol{u}}_i$, we then introduce a regularization loss:
$$
L^i_u=\left\|\check{\boldsymbol{u}}_i-\hat{\boldsymbol{u}}_i\right\|^2.
$$
The total loss function for training the CBF network is:
$$
L=\mathbf{L}_{CBF}+\sigma \sum_i L^i_u
$$
The refined control input $\check{\boldsymbol{u}}_i$ will be used in practice to simultaneously keep formation (MPC-CLF) while avoiding obstacles (SAPIE-CBF).
$\sigma$ is a weighting factor balancing safety and control performance.

To conclude Sections~\ref{sec:mpc-clf}, ~\ref{sec:sapie} and~\ref{sec:nn-cbf}, we present the training pipeline in Fig.~\ref{fig:framework}.






\subsection{Deadlock Resolution}
Another challenge in real-world deployment is deadlocks in quadrupedal robot formation control that arise when restrictive CBF constraints conflict with the MPC-CLF framework's control input, preventing motion in high-density environments or constrained spaces where mutual constraints make movement infeasible without violating safety conditions. For example, in practice, when multiple quadrupedal robots navigate a narrow passage, the restrictive CBF constraints may prevent any robot from moving forward despite feasible nominal commands, leading to a standstill. Given the tight coupling between motion control and gait coordination in quadrupedal locomotion, directly modifying the control input may destabilize the gait or disrupt step transitions.  To address this, we propose an event-triggered deadlock resolution mechanism that adjusts the nominal state, applying a controlled perturbation through transformation matrix $T$ instead of the control input, ensuring that the stability constraints of MPC-CLF remain intact while satisfying the safety constraints of CBF to restore motion feasibility.
Formally, we define the event-triggering deadlock indicator as:
$$
E_i=\operatorname{sign}\left(\left\|J_{\mathrm{CLFs}}-L_{\mathrm{CBFs}}\right\|+\left\|\boldsymbol{v}_i\right\|\right)-\operatorname{sign}\left(J_{\mathrm{CLFs}}\right)
$$
where a deadlock is detected when $E_i<0$, triggering an adjustment to the nominal state rather than perturbing the control input, ensuring MPC-CLF constraints remain intact. Upon the deadlock detection being triggered, the nominal state $\hat{\boldsymbol{x}}_i=\frac{\sum_{j=1, j \neq i}^n c_{i j}\left(x_j+\Delta_{ij}\right)+s_i\left(x_r+\Delta_{i r}\right)}{\sum_{j=1, j \neq i}^n c_{i j}+s_i}$ in \eqref{Eq.nominalstate} is modified as:
$
\tilde{\boldsymbol{x}}_i=T\left(\hat{\boldsymbol{x}}_i-\boldsymbol{x}_i\right)+\boldsymbol{x}_i
$
The transformation matrix $T$ applies a controlled perturbation to shift the nominal reference, ensuring a nonzero control input.
\begin{figure}
	\centering
	\includegraphics[width=\columnwidth]{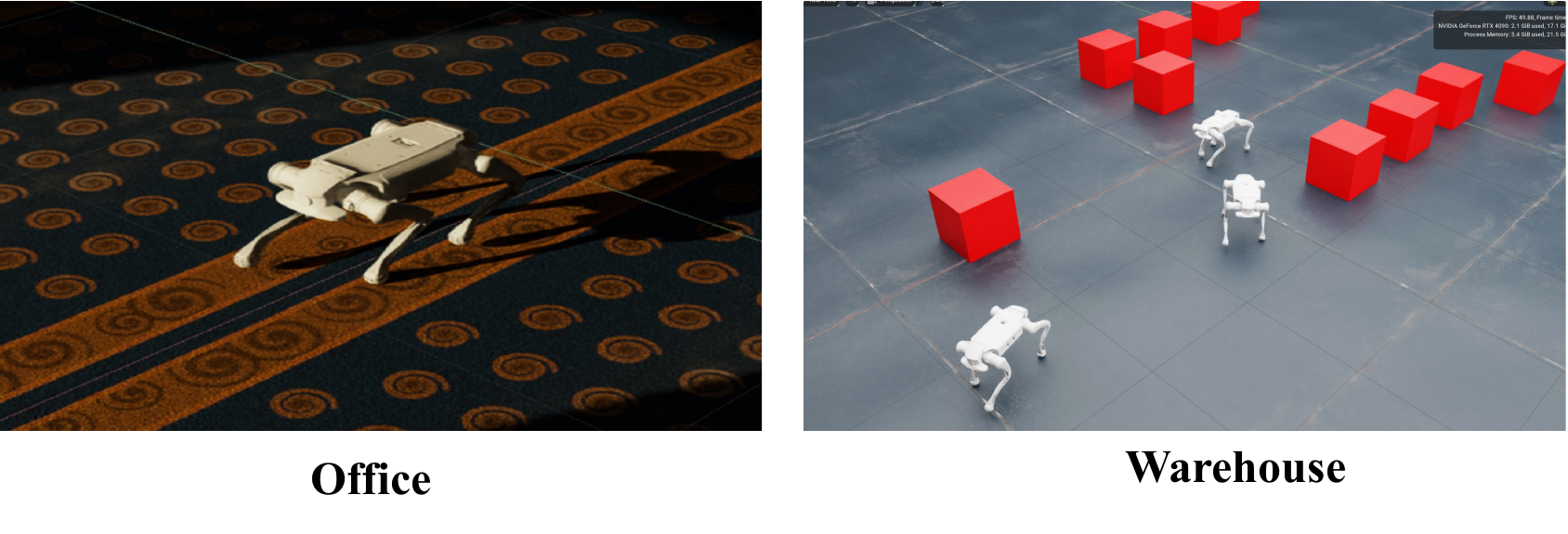}
 \vspace{-23pt}
	\caption{The digital twin of XG, an autonomous quadrupedal robot, shown operating in two simulated environments (Office and Warehouse). These high-fidelity simulations facilitate robust development and Sim2Real transfer.}
	\label{fig:twin}
 \vspace{-7pt}
\end{figure}
If a deadlock occurs $\left(E_i<0\right)$, the event-triggered mechanism adjusts the nominal state $\hat{\boldsymbol{x}}_i$ to $\tilde{\boldsymbol{x}}_i$, thereby modifying the nominal control input $\hat{\boldsymbol{u}}_i$. As long as at least one feasible component of $\hat{\boldsymbol{u}}_i$ remains nonzero, the actual control input $\check{\boldsymbol{u}}_i$ will also be nonzero in the subsequent step, effectively resolving the deadlock and restoring motion feasibility while preserving the stability constraints of the MPC-CLF framework. Notably, this mechanism is triggered only once per deadlock occurrence.

\section{ SIMULATIONS AND EXPERIMENTS}

\subsection{System Design}
\begin{figure}
	\centering
	\includegraphics[width=\columnwidth]{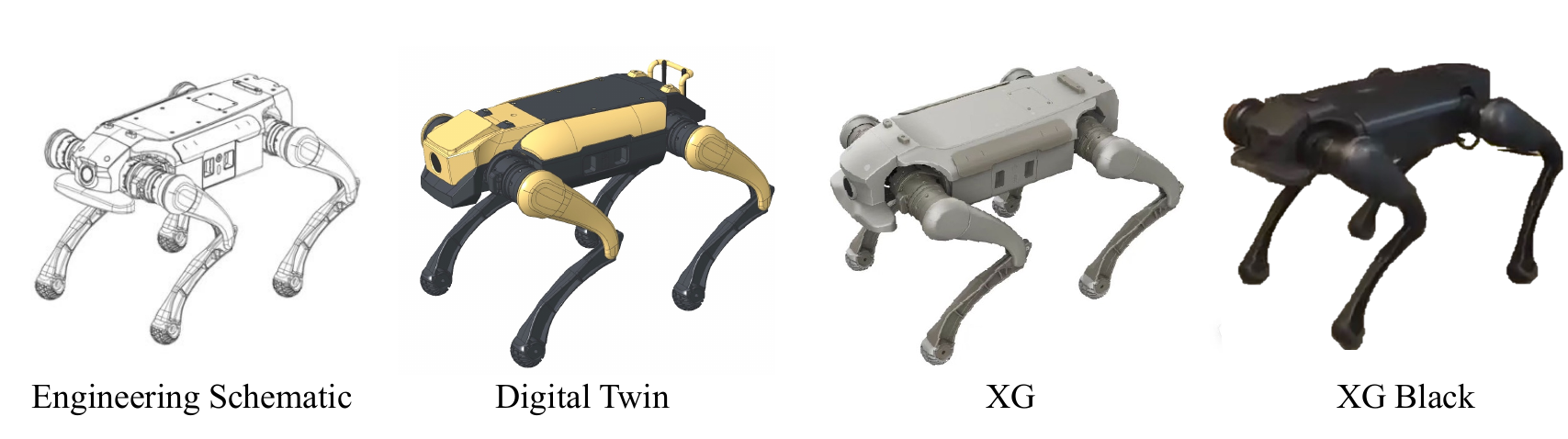}
    \vspace{-0.8cm}
	\caption{Evolution of the XG quadrupedal robot platform. From left to right: Engineering schematic showing the mechanical design and joint placement; Digital twin model implemented in NVIDIA Isaac Sim for high-fidelity simulation; Standard XG robot hardware implementation; XG Black variant with modified chassis configuration for enhanced payload capacity and experimental deployments.}
	\label{fig:dogs}
 \vspace{-12pt}
\end{figure}

\begin{figure}
	\centering
	\includegraphics[width=\columnwidth]{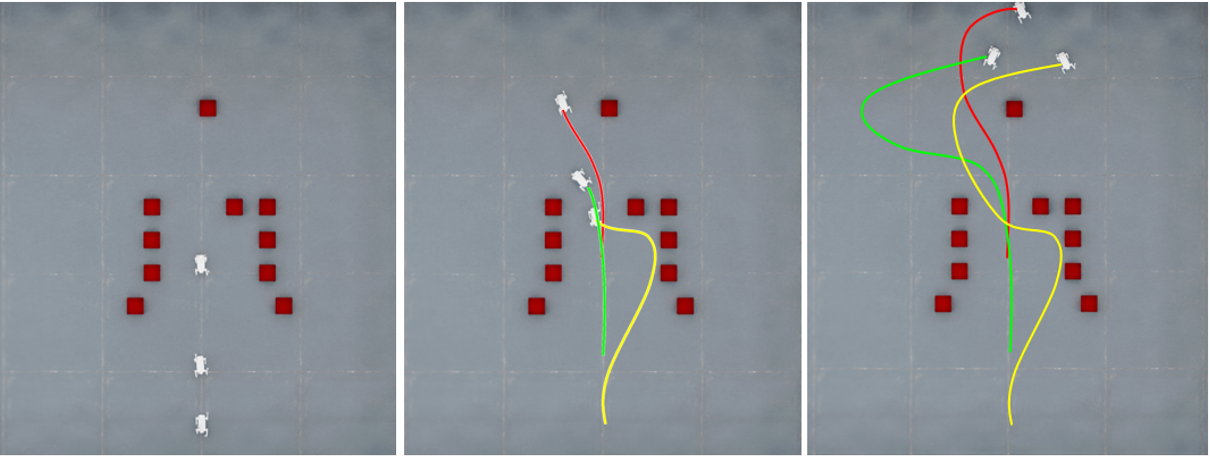}
    \vspace{-0.4cm}
	\caption{ Snapshots of the simulation environment showing formation evolution over time. Left: Initial configuration with five quadrupedal robots (white) and static obstacles (red cubes) arranged in a corridor formation. Middle: Intermediate state showing robots navigating through the obstacle field with color-coded trajectories (yellow, green, red) indicating individual robot paths. Right: Final configuration demonstrating successful passage through the obstacle field while maintaining inter-robot coordination. }
	\label{fig6}
 \vspace{-12pt}
\end{figure}

\subsubsection{Quadruped robot}

The quadruped robot platform features 12 degrees-of-freedom (DoF) and incorporates an onboard battery pack with Wi-Fi connectivity for autonomous field operations. The sensing suite integrates a Mid-360 LiDAR system with 40-beam configuration, offering a $59^{\circ}$ vertical field of view (FoV) and complete $360^{\circ}$ horizontal coverage below the horizon, effectively eliminating perceptual blind spots. Computational resources comprise dual RK3588-based processing boards ($4\times$ Cortex-A76 cores at 2.4 GHz paired with $4\times$ Cortex-A55 cores at 1.8 GHz) with 16 GB of RAM. These boards are dedicated to the planning and locomotion control algorithms, respectively, ensuring real-time performance.
Fig.~\ref{fig:dogs} presents the system's evolution from engineering schematic (left most), through simulation digital twin environment (second from left), to the fully realized XG robot hardware implementations (right two variants) with different chassis configurations and onboard equipment.

\subsubsection{Digital Twins in NVIDIA Omniverse Isaac Sim}
To facilitate development and validation, we created high-fidelity digital twins in NVIDIA Omniverse Isaac Sim. These simulations precisely replicate the physical properties of real-world environments, with our quadruped robot USD models successfully integrated into the simulator. All critical parameters—including dimensions, mass distribution, actuator torque limits, and joint constraints—were meticulously calibrated to match their hardware counterparts, substantially minimizing the sim-to-real gap. We initially leveraged this environment to train locomotion gaits through reinforcement learning (deployed successfully on physical robots, though details are beyond this paper's scope). Subsequently, we extended our simulations to multi-robot scenarios across diverse environments including warehouses and office spaces, as shown in Fig. \ref{fig:twin}. The virtual testing environments were constructed with geometric fidelity to our physical testbeds, incorporating identical obstacle placements and dimensional constraints. This comprehensive simulation framework enabled thorough validation of our formation control and collision avoidance algorithms in obstacle-dense environments before real-world deployment, significantly accelerating the development cycle while reducing hardware risks.
\begin{figure}
	\centering
	\includegraphics[width=\columnwidth]{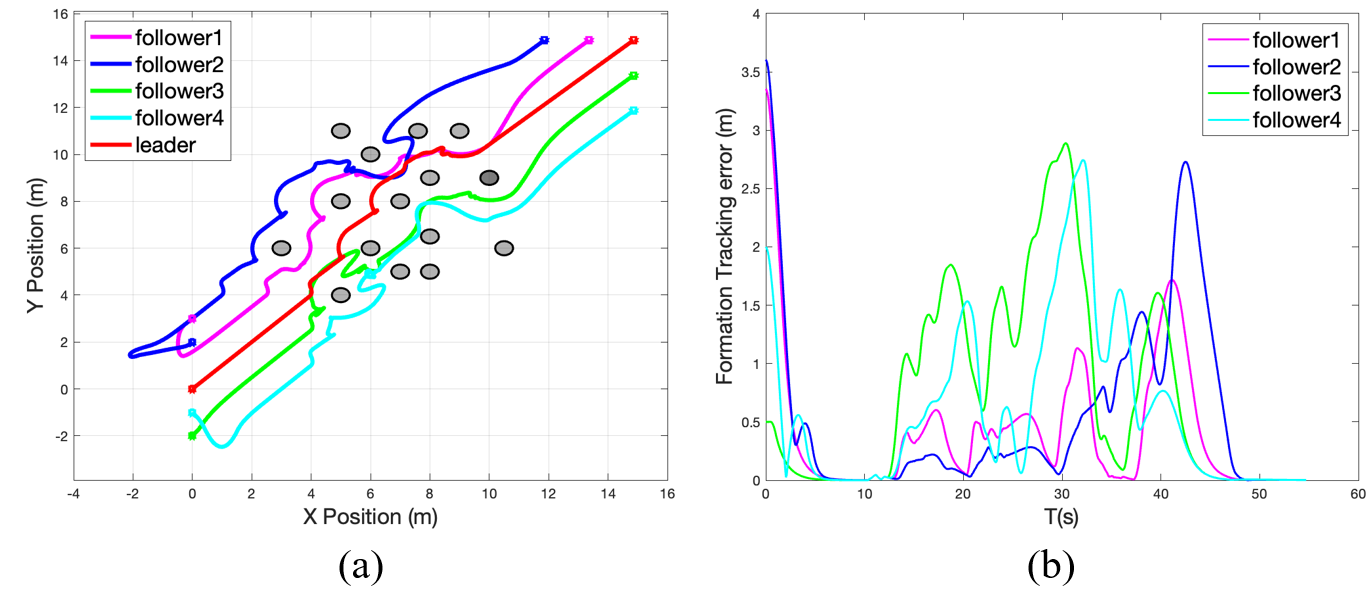}
    \vspace{-0.8cm}
	\caption{Quantitative analysis of formation control in obstacle-dense environments. (a) Planar trajectory plot showing the complete paths of all robots (leader in red, followers in magenta, blue, green, and cyan) navigating through randomly placed obstacles (gray circles). The consistent spacing between trajectories after obstacle negotiation demonstrates successful formation maintenance. (b) Formation tracking error versus time for all follower robots, showing initial error accumulation during obstacle avoidance 
    (peaks at $t\approx 30s$) and subsequent convergence as robots establish stable formation after clearing the obstacle field. The temporary increase in tracking error corresponds directly to obstacle negotiation phases, validating the controller's ability to prioritize collision avoidance while keeping formation whenever feasible.}
	\label{fig7}
 \vspace{-18pt}
\end{figure}
\vspace{-10pt}
\subsection{Simulated Experiments}
\vspace{-5pt}
\subsubsection{Simulation Setup}
The proposed planner has been fully trained within NVIDIA Omniverse and evaluated using our designed quadrupedal robot XG with a pre-trained RL locomotion policy provided by the Orbit Framework \cite{mittal2023orbit}. Our algorithm operates under the ROS2 Humble operating system, an open-source framework widely used for developing robotics and automation systems \cite{macenski2022robot}.

Our hierarchical control architecture is implemented and trained within this simulated environment. As shown in Fig.~\ref{fig6}, we selected the NVIDIA warehouse environment (approximately 25 × 35m) where cuboid obstacles form narrow corridors. Our experiments focus on formation maintenance, collision avoidance, and deadlock resolution in these challenging scenarios.

\subsubsection{Scenario 1: Safe Navigation in Obstacle Dense Environment}\label{sec:Scenario 1}
We evaluated our algorithm across multiple simulation and hardware environments. In the first simulation environment, we deployed $n=5$ quadruped robots operating in a 2D plane with static obstacles (Fig.~\ref{fig:dogs}). Each robot was modeled as a point mass for planning purposes. The CBF implementation utilized a convolutional neural network architecture (4 layers, 64, 128, 64, 1) integrated with our SAPIE encoding mechanism (processing a maximum of 4 neighbors). To ensure robustness, we sampled diverse obstacle configurations and neighbor arrangements during the training phase, generating a comprehensive dataset of safe and unsafe states.

The MPC optimization problem in \eqref{Eq. 3} was solved using the Interior Point Optimizer (IPOPT) \cite{biegler2009large}, an open-source solver specifically designed for nonsmooth, continuous nonlinear optimization challenges. Fig. \ref{fig3} depicts the quadruped robot formation and communication network topology employed throughout our simulations.

Results are presented in Fig. \ref{fig7}, with Fig. \ref{fig7}(a) illustrating the motion trajectories of all five quadruped robots navigating through multiple obstacles positioned near their starting locations. Fig. \ref{fig7}(b) displays the corresponding formation tracking error using our proposed approach.

The trajectory plot in Fig. \ref{fig7}(a) demonstrates that our method successfully guided the quadruped robots to establish and keeping formation while navigating around the positioned obstacles. Initially, the robots operated without a defined formation before encountering the obstacles. During obstacle traversal, they progressively established their formation structure, significantly reducing tracking error. As shown in Fig. \ref{fig7}(b), after clearing the obstacle field, the quadrupedal robots achieved stable formation with tracking errors converging toward minimal values, confirming the effectiveness of our approach in complex environments.
\begin{figure}[t]
	\centering
	\includegraphics[width=0.5\columnwidth]{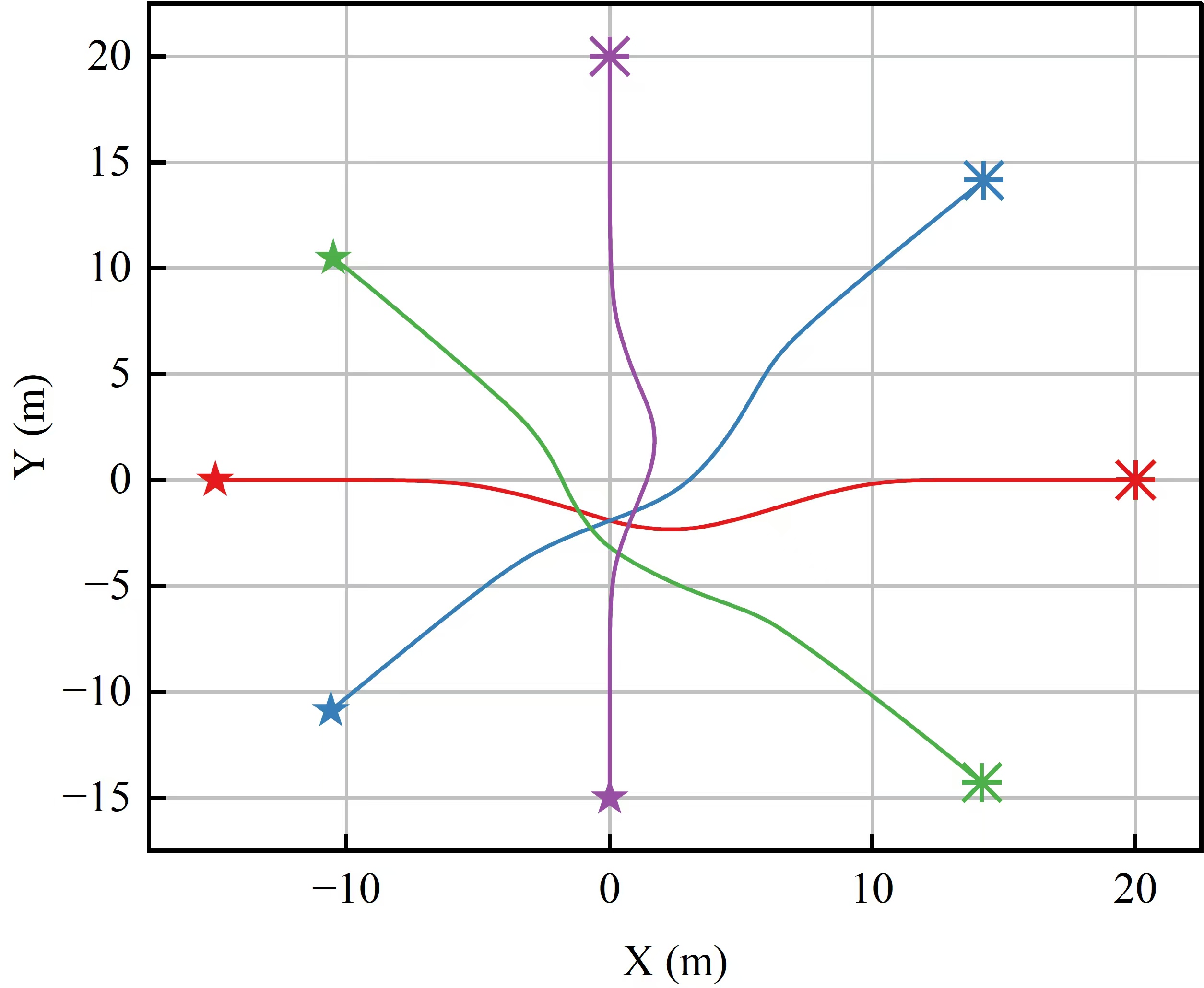}
    \vspace{-10pt}
	\caption{Verification of dynamic collision avoidance among four quadruped robots. The asterisks indicate the starting positions, while the stars mark the final destinations. Each colored trajectory (red, blue, green, and purple) shows the navigation path, with all robots successfully resolving potential deadlocks at the intersection point near the origin.}
	\label{fig8}
 \vspace{-18pt}
\end{figure}



\subsubsection{Scenario 2: Deadlock Resolution and Dynamic Collision Avoidance}
In the second scenario, we evaluated our system's performance by orchestrating a challenging interaction where four groups of quadruped robots simultaneously moved toward each other, creating conditions for both collision risk and potential deadlock situations. As illustrated in Fig. \ref{fig8}, our approach effectively ensured safe inter-robot interactions despite the complex crossing trajectories. When deadlock conditions emerged, the quadrupedal robots autonomously activated our proposed resolution mechanism without requiring external intervention. This experimental outcome substantiates the effectiveness of our algorithm in both detecting and resolving deadlock conditions, thereby enhancing the robustness of multi-robot navigation in confined spaces.

To evaluate the adaptability of our method across different environments, we leveraged digital twins in NVIDIA Isaac Sim that accurately replicated real-world physical properties and deployed the learned policy on the quadrupedal robots (Fig.~\ref{fig6}). 
The communication graph between locally interacting robots is depicted in Fig. \ref{fig6}(a), illustrating the distributed nature of our approach. In this challenging scenario, three quadrupedal robots successfully navigated through a narrow corridor while maintaining their desired formation, avoiding both static and dynamic obstacles, and efficiently reaching their predefined waypoints, thereby validating the effectiveness of our integrated control framework.

\begin{figure*}[hbt!]
    \centering
    \includegraphics[width=\linewidth]{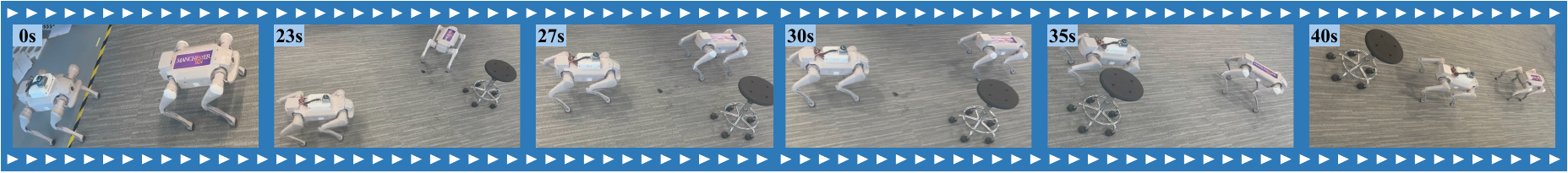}
    \vspace{-14pt}
    \caption{Verification of dynamic collision avoidance among two quadruped robots in a real-world environment, captured at six different timestamps (0 s, 23s, 27s, 30s, 35s, and 40s). The sequence shows our XG quadrupedal robots navigating around a dynamic obstacle (office chair) while maintaining formation. Initially (0s), the robots begin in a line formation with the leader carrying a purple marker. As the sequence progresses, they detect the obstacle, execute real-time trajectory adjustments (23s–35s), and successfully navigate past it while preserving their relative positions. By the final frame (40s), both robots have safely negotiated the obstacle and resumed their desired formation, demonstrating the effectiveness of our DDS-based communication architecture and distributed control framework in physical environments.}
    \label{fig9}
    \vspace{-18pt}
\end{figure*}
\vspace{-7pt}
\subsection{Real-world Experiments}
\vspace{-5pt}
\subsubsection{Formation Control with XG Platform}
We validated the proposed algorithm with a team of two XG quadrupedal robots performing formation control and collision avoidance tasks. These robots are equipped with IMU and LiDAR sensors and utilize Data Distribution Service (DDS) for inter-robot communication. Additionally, we developed a dedicated middleware based on DDS that provides a unified application programming interface (API) across different robot platforms to optimize data synchronization and scalability, establishing a foundation for future large-scale swarm control applications.

The decentralized architecture, efficient Quality of Service (QoS) mechanisms, and low-latency characteristics of DDS enable real-time state sharing, enhancing coordination capabilities for both formation maintenance and obstacle avoidance. As shown in Fig. \ref{fig9}, our controller successfully navigated around dynamic obstacles while maintaining the desired formation between robots throughout the experiment.

\subsubsection{Cross-Platform Validation}
We further demonstrated the generalizability of our approach by deploying the same control strategy on a second pair of quadrupedal robots (\ref{fig:dogs}) with different physical characteristics. These robots, designed for heavy payload transportation, feature different materials and structural designs to accommodate higher load-bearing capacity and enhance stability. Despite these significant physical differences from the original XG platform, our formation control framework performed robustly without requiring model recalibration, successfully maintaining desired formations and avoiding obstacles. This cross-platform validation highlights the generalizability of our approach, suggesting its potential applicability across heterogeneous quadrupedal platforms with minimal adaptation requirements.

\section{Conclusion and Future Work}
\label{sec:5}
This paper presents a distributed formation control framework for quadrupedal robots in dynamic environments. By integrating DMPC with CLFs for stability and a neural network-based CBF with SAPIE encoding for decentralized safety, our approach ensures formation maintenance, collision avoidance, and real-time feasibility. A low-latency DDS-based communication system enhances coordination, while an event-triggered mechanism resolves deadlocks in constrained spaces. Extensive simulations in NVIDIA Omniverse Isaac Sim and real-world experiments on our XG quadrupeds validate its robustness in complex scenarios. Future work includes extending to multi-agent visual-LiDAR SLAM and potential open-source release.

\vspace{-1pt}

\bibliography{main}
\bibliographystyle{IEEEtran}
%
%
%
%

\end{document}